\documentclass[12pt,a4paper]{article}
\usepackage{amsmath,amssymb,amsthm}

\numberwithin{equation}{section}

\newtheorem{theorem}{Theorem}[section]
\newtheorem{lemma}[theorem]{Lemma}

\theoremstyle{definition}
\newtheorem{remark}[theorem]{Remark}

\DeclareMathOperator{\real}{Re}

\begin{document}
\title{A solvable model of the breakdown of the adiabatic approximation}

\author{
A. Galtbayar\thanks{
Department of Applied Mathematics, National University of Mongolia, 
University Street 3, Ulaanbaatar, 21-046, Mongolia. 
}\,, 
A. Jensen\thanks{
Department of Mathematical Sciences, Aalborg University, Skjernvej 4A, 
DK-9220 Aalborg \O{}, Denmark} \,,
and  
K. Yajima\thanks{
Department of Mathematics, Gakushuin University, 
1-5-1 Mejiro, Toshima-ku, Tokyo 171-8588, Japan. 
} 
}

%
%
%
\maketitle 
\begin{abstract}
Let $L\geq0$ and $0<\varepsilon\ll1$. Consider the following time-dependent  family of $1D$ Schr\"{o}dinger equations with scaled harmonic oscillator potentials
$
i\varepsilon\partial_t u_{\varepsilon}=-\tfrac12\partial_x^2u_{\varepsilon}+V(t,x)u_{\varepsilon}$,
$u_{\varepsilon}(-L-1,x)=\pi^{-1/4}\exp(-x^2/2)
$,
where
$
V(t,x)=
(t+L)^2x^2/2$, $t<-L$,
$
V(t,x)=
0$,  $-L\leq t \leq L$, and
$
V(t,x)=(t-L)^2x^2/2$, $t>L$.
The initial value problem is explicitly solvable in terms of Bessel functions. Using the explicit solutions we show that the adiabatic theorem breaks down as $\varepsilon\to 0$. 
For the case $L=0$ complete results are obtained. The survival probability of the ground state $\pi^{-1/4}\exp(-x^2/2)$ at microscopic time $t=1/\varepsilon$ is $1/\sqrt{2}+O(\varepsilon)$. For $L>0$ the framework for further computations and preliminary results are given.
\end{abstract}

\section{Introduction} 

Let $\mathcal{H}$ be a Hilbert space and $\{H(t) \colon -a<t<a\}$ 
a family of selfadjoint operators in $\mathcal{H}$. Suppose that the
time dependent Schr\"odinger equation 
\begin{equation} \label{adiabatic-t}
i\varepsilon\partial_t u(t) = H(t) u(t)
\end{equation} 
with a small parameter $0<\varepsilon\ll 1$ generates a unique unitary 
propagator $U_\varepsilon(t,s)$ and that $t \mapsto (H(t)-i)^{-1}\in \mathbf{B}(\mathcal{H})$ 
is of class $PC^2(\mathbb{R})$, i.e. piecewise $C^2$. 
Suppose further that $H(t)$ has an  isolated simple 
eigenvalue $\lambda(t)$ with  an associated 
normalized eigenfunction $\varphi(t)$, both of class 
$PC^1$ for $t \in (-a,a)$, such that 
\begin{equation*}
H(t) \varphi(t) = \lambda(t) \varphi(t), \ \ \|\varphi(t)\|^2 = 1, \quad -a<t<a.
\end{equation*}
Then the classical theorem of adiabatic approximation due 
to Born-Fock\cite{BF} and Kato\cite{Ka}  
implies that the solution $u_\varepsilon(t)=U_\varepsilon(t,0)\varphi(0)$ of the initial value 
problem: 
\begin{equation} \label{adiabatic-t0}
i\varepsilon\partial_t u_\varepsilon(t) = H(t) u_\varepsilon(t), \quad u_\varepsilon(0)= \varphi(0), 
\end{equation} 
with a small parameter $0<\varepsilon\ll1$ satisfies for a $\delta<a$
\begin{equation} \label{ad-t}
\|u_\varepsilon(t) -e^{-i\varepsilon^{-1}\int_0^t \lambda(s) ds}\varphi(t)\| \leq C_\delta \varepsilon, \quad |t|<\delta,
\end{equation} 
where $\|\cdot \|$ is the norm of $L^2(\mathbb{R})$. More precisely, a gap condition is imposed, i.e. assume that
\begin{equation*}
\inf\bigl\{\textrm{dist}\bigl(\{\lambda(t)\},
\sigma(H(t))\setminus\{\lambda(t)\}\bigr):\,-\delta<t<\delta
\bigr\}>0.
\end{equation*}
Furthermore, the phase of $\varphi(t)$ has to be fixed correctly. Let
\begin{equation*}
P(t)=-\frac{1}{2\pi i}\int_{|z-\lambda(t)|=\eta}
(H(t)-z)^{-1}dz
\end{equation*}
be the (Riesz) projection onto the eigenspace $\mathrm{Ker}(H(t)-\lambda(t))$. Then for a sufficiently small $\delta>0$ one defines for $-\delta<t<\delta$
\begin{equation*}
\varphi(t)=\frac{P(t)\varphi(0)}{\|P(t)\varphi(0)\|}.
\end{equation*}
With this choice the result \eqref{ad-t} holds.

This adiabatic theorem has been substantially 
elaborated and extended to more general situations, and it has been widely 
applied in various 
fields of mathematical physics, see e.g. Teufel's monograph \cite{Teufel} 
and the references therein.

\noindent
\textbf{The eigenvalue dives into the continuum.} 
We consider the situation that eigenvalue $\lambda(t)$ dives into the 
continuous spectrum of $H(t)$ at, say, $t=-L>-a$, stays in the continuum 
of $H(t)$ for $-L \leq t \leq L$, and comes out again for $t>L$ as an 
isolated eigenvalue of $H(t)$. Under the assumption 
that $\lambda(t)$ remains as an (embedded) eigenvalue of $H(t)$ 
for $-L\leq t \leq L$, then a general argument has been established 
and a result similar to \eqref{ad-t} is obtained (see Teufel\cite{Teufel}). 
Moreover, the result has been applied by D\"urr-Pickl \cite{DP} 
to the Dirac equation to explain the adiabatic pair creation and 
by Cornean-Jensen-Kn\"{o}rr-Nenciu \cite{CJK} 
to specific finite rank perturbations of Schr\"odinger equations. 
However, if $H(t)$ has no embedded eigenvalues for $-L\leq t \leq L$ 
and the eigenvalue $\lambda(t)$ ``melts away into the continuum", then there is no 
general theory to deal with the problem; it is even not clear what is meant 
by the adiabatic approximation. We should mention that embedded eigenvalues 
in the continuum 
are very unstable under a perturbation and for genuinely time dependent 
Hamiltonians embedded eigenvalues would hardly persist 
for any finite time interval. 

\noindent
\textbf{Harmonic oscillators which become the free Hamiltonian.}
To understand these phenomena we study an explicitly solvable model.
More precisely, we study the solution of the Schr\"odinger equation 
which can be written in terms of the macroscopic time variable as 
\begin{equation} \label{harmonic-a}
i\varepsilon \partial_t u_\varepsilon = -\frac12 \partial_x^2 u_\varepsilon + V(t, x) u_\varepsilon, \quad 
u_\varepsilon(-L-1, x) = \varphi_0(x),  
\end{equation} 
which is a scaled harmonic oscillator 
for $t <-L$ and $t>L$ and $V(t,x)=0$ for $-L\leq t \leq L$:
\begin{equation}\label{macro}
V(t,x)=\begin{cases}
(t+L)^2 x^2/2, & t < - L,\\
0, & -L\leq t\leq L,\\
(t-L)^2 x^2/2, & t > L,
\end{cases}
\end{equation}    
and the initial state $\varphi_0(x)=\pi^{-\frac14}e^{-x^2/2}$ is the 
normalized ground state of the initial Hamiltonian 
$H(-L-1)=-(1/2)\partial_x^2 + (1/2)x^2$. 
We are particularly interested in the asymptotic behavior as $\varepsilon \to 0$ 
of $u_\varepsilon(t,x)$ at $t=L+1$ when $H(t)$ again becomes $-(1/2)\partial_x^2 + (1/2)x^2$. 

It is well known that the equation \eqref{harmonic-a} generates a unique 
unitary propagator $\{U_\varepsilon (t, s)\colon -\infty<t,s<\infty\}$ which 
is simultaneously an isomorphism of $\mathcal{S}(\mathbb{R})$ and of 
$\Sigma(2n)$, $n=0,1,\dots$, the domain of 
$(-(1/2)\partial_x^2 + (1/2)x^2)^n$. For $\varphi \in \Sigma(2)$, 
$\mathbb{R} \times \mathbb{R} \ni (t,s) \mapsto U_\varepsilon (t,s)\varphi\in L^2(\mathbb{R})$  is $C^1$ in 
$(t,s)$ and $u_\varepsilon(t)= U_\varepsilon(t, -L-1)\varphi $ (see Fujiwara\cite{Fu}).
We should emphasize, however, that $H(t)$ fails to satisfy the 
assumptions of the theory of adiabatic approximation in 
two ways: (1) all eigenvalues dive into continuum simultaneously; 
(2) the domain of $H(t)$ has a sharp transition at time 
$t=-L$ and $t=L$ and the resolvent $(H(t)-i)^{-1}$ is not of 
class $C^1$ at these points.  

We shall study \eqref{harmonic-a} in the microscopic 
time variable, viz. we change the time variable to $s= t/\varepsilon$ and study 
$v_\varepsilon (s,x)= u_\varepsilon ({\varepsilon}s, x)$. $v_\varepsilon(s,x)$ satisfies 
\begin{equation} \label{harmonic}
i\partial_s v_\varepsilon = -\frac12 \partial_x^2 v_\varepsilon + V(\varepsilon s, x) v_\varepsilon\,
\ \  v_\varepsilon(-{\varepsilon}^{-1}(L+1),x) = \varphi_0(x) 
\end{equation} 
and, as we only consider \eqref{harmonic} in what follows 
{\it we denote the microscopic time variable again by $t$ instead of $s$}. 
Our result will be rather complete in the case $L=0$, however, when 
$L>0$, the situation becomes exceedingly complicated and we have to 
be satisfied with  partial results which should be considered as the starting 
point for  further study.  

\noindent
\textbf{Summary of results.}
The main results in the case $L=0$ are stated in Theorem~\ref{1}. Let $v_{\varepsilon}(t,x)$ denote the solution to \eqref{harmonic} with initial state $\varphi_0(x)=\pi^{-\frac14}e^{-x^2/2}$ (at time $t=-1/\varepsilon$). Then at time $t=1/\varepsilon$ we have
\begin{equation*}
v_{\varepsilon}(1/\varepsilon,x)=m_{\varepsilon,0}(1/\varepsilon)
e^{-l^{\ast}_{\varepsilon,0}(1/\varepsilon)x^2/2}+O(\varepsilon)
\end{equation*}
as $\varepsilon\to0$. Here $m_{\varepsilon,0}(1/\varepsilon)$ and 
$l^{\ast}_{\varepsilon,0}(1/\varepsilon)$ are given by \eqref{m0} and \eqref{l0}, respectively. Note that these coefficients are highly oscillatory as $\varepsilon \to 0$, exhibiting the breakdown of the adiabatic approximation.

This result allows one to compute the survival probability of the time $t=-1/\varepsilon$ initial state $\varphi_0$ to time $t=1/\varepsilon$. The result is
\begin{equation*}
|\langle v_{\varepsilon}(1/\varepsilon,x),\varphi_0(x)\rangle|^2
=\frac{1}{\sqrt{2}}+O(\varepsilon)
\end{equation*}
as $\varepsilon\to0$. Thus the initial state survives with a positive probability which is less than $1$. This survival probability was computed in Bachmann et al.\cite{BFG} by different methods.

The results in the case $L>0$ are stated in Theorem~\ref{thm2}. These partial results are somewhat complicated to state. Roughly, there exist sequences $\varepsilon_n\to0$ as $n\to\infty$ such that there is a positive survival probability of the initial state, which however rapidly tends to zero as the length of (microscopic) time $2L/\varepsilon$ spent in the continuum increases.

\section{The case $L=0$}
We first consider the case $L=0$, viz. the case where the eigenvalues of $H(t)$ 
touch upon the continuum only at time $t=0$, but all simultaneously. 
We should mention that the 
problem for this case has been studied by Bachmann~et~al.\cite{BFG} 
by a  method very different from ours and the results slightly overlap. 
 
We record a few lemmas which we shall use in what follows. The first one 
can be found in Yajima\cite{Ya}.  

\begin{lemma} Let $l_{\varepsilon}(t)$ be the solution of the Riccati equation 
\begin{equation} \label{Riccati-eq}
l_{\varepsilon}'(t) + i l_{\varepsilon}(t)^2 = i \varepsilon^2 t^2
\end{equation} 
with initial condition  
\begin{equation} \label{IC}
l_{\varepsilon}(-1/\varepsilon)=1.
\end{equation} 
Suppose that $m_{\varepsilon}(t)$ solves  
\begin{equation}
\label{m}
i m_{\varepsilon}'(t) = \frac12 m_{\varepsilon}(t) l_{\varepsilon}(t), 
\quad m_{\varepsilon}(-1/\varepsilon)=\pi^{-\frac14}. 
\end{equation} 
Then, $|m_{\varepsilon}(t)|^4=\pi^{-1}\real l_{\varepsilon}(t)$ and  
\begin{equation} \label{sol}
v_\varepsilon(t,x)= m_{\varepsilon}(t) e^{-l_{\varepsilon}(t)x^2/2} 
\end{equation} 
is the solution of the initial value problem for the Schr\"odinger equation 
\begin{equation} 
i\partial_t v_\varepsilon =- (1/2)\partial^2_x v_\varepsilon +  (t^2 \varepsilon^2 x^2/2)v_\varepsilon, \quad 
v_\varepsilon(-1/\varepsilon,x)= \pi^{-\frac14}e^{-\frac{x^2}{2}}.
\end{equation}  
\end{lemma} 

\noindent
\textbf{General solutions of the Riccati equation.} 
Bessel functions of the first kind $J_\nu(z)$ and the second kind $Y_\nu(z)$ 
are defined by 
\begin{gather} \label{Bessel-def}
J_\nu (z)= \left(\frac{z}{2}\right)^\nu 
\sum_{k=0}^\infty (-1)^k 
\frac{(z^2/4)^k}{k!\,\Gamma(\nu+k+1)}, \\
Y_\nu(z)= \frac{J_\nu(z)\cos\,\nu{\pi} - J_{-\nu}(z)}{\sin\,{\nu}\pi}.
\label{Bessel-def-a}
\end{gather} 
They are linearly independent solutions of Bessel's equation 
\[
z^2 J''_{\nu}(z)+ z J'_\nu(z)+ (z^2-\nu^2)J_{\nu}(z)=0 
\]
and their positive zeros are interlaced (see DLMF\cite{DLMF} (10.21.3)).  

\begin{lemma} \label{hoho}
Let $\varepsilon>0$ and  $\kappa\in (\mathbb{C}\cup\{\infty\})
\setminus\mathbb{R}$.
Define $w(s,\kappa)$ for $s\geq 0$ by  
\begin{equation} \label{w-def}
w(s,\kappa)= 
s^{1/8}
\big(- J_{-\frac14}(2\sqrt{s}) +{\kappa} J_{\frac14}(2\sqrt{s})\big)
\end{equation} 
where the principal branches are assumed for the Bessel functions, and in the case $\kappa=\infty$ the first term is omitted. 
Define 
\begin{equation} \label{wtil-def}
\tilde{w}_{\varepsilon}(t,\kappa)= w\bigl(\frac{\varepsilon^2 t^4}{16},\kappa\bigr), \quad t>0.
\end{equation} 
Then, $\tilde{w}_{\varepsilon}(t,\kappa)$ may be analytically continued to an entire function 
of $t \in \mathbb{C}$ and it does not vanish on the real line. 
\end{lemma} 
\begin{proof} From the definition of Bessel functions \eqref{Bessel-def}, we have 
\begin{equation} \label{Bessel-expansion}
J_\nu (2\sqrt{s})= s^{\nu/2} M_\nu(s), \quad 
M_\nu(s)= \sum_{k=0}^\infty (-1)^k 
\frac{s^k}{k!\,\Gamma(\nu+k+1)} ,
\end{equation} 
and $M_{\nu}(s)$ is evidently an entire function of $s \in \mathbb{C}$. 
It follows that 
\begin{equation} \label{M1/4}
{w}(s,\kappa)= -M_{-\frac14}(s) + \kappa s^{\frac14}M_{\frac14}(s) 
\end{equation} 
and 
\begin{equation}  \label{winM}
\tilde{w}_{\varepsilon}(t,\kappa)= -M_{-\frac14}({\varepsilon^2 t^4}/{16}) + (\kappa \sqrt{\varepsilon}t/2)
M_{\frac14}({\varepsilon^2 t^4}/{16}) 
\end{equation} 
is an entire function of $t\in \mathbb{C}$. As $w(s,\kappa)$ is a linear combination 
of $J_{\frac14}(2\sqrt{s})$ and $J_{-\frac14}(2\sqrt{s})$ with 
non $\mathbb{R}$-related coefficients, $\tilde{w}_{\varepsilon}(t,\kappa)\not=0$ for $t>0$. 
But \eqref{winM} shows $\tilde{w}_{\varepsilon}(-t,\kappa)= \tilde{w}_{\varepsilon}(t,-\kappa)$ 
and the same is true for $t <0$ and, 
$\tilde{w}_{\varepsilon}(0,\kappa)= -M_{-\frac14}(0)= -\Gamma(3/4)^{-1}\not=0$.
This completes the proof. 
\end{proof} 

\begin{lemma} Let $\varepsilon\not=0$,
$\kappa\in (\mathbb{C}\cup\{\infty\})
\setminus\mathbb{R}$,
and $s= {\varepsilon^2 t^4}/{16}$. Let $w(s,\kappa)$ and $\tilde{w}_{\varepsilon}(t,\kappa)$ be 
as in Lemma~\ref{hoho}. Then: \\[7pt]
{\rm (1)} With $\kappa$ being an arbitrary constant, the general solution
of the Riccati equation \eqref{Riccati-eq} is given by 
\begin{equation} \label{w}
l_{\varepsilon}(t,\kappa)  = \frac{-4i s w'(s,\kappa)}{t w(s,\kappa)} 
 =- i \frac{\tilde{w}_{\varepsilon}'(t,\kappa)}{\tilde{w}_{\varepsilon}(t,\kappa)}\,. 
\end{equation}  
It is a holomorphic function of $t$ in a complex 
neighborhood of the real line. \\[7pt]
{\rm (2)} We may express $l_\varepsilon(t,\kappa)$ without using derivatives:  
\begin{align}
l_{\varepsilon}(t,\kappa) & =\frac{-4is^\frac12
\big({\kappa}J_{-\frac34}(2\sqrt{s}) + J_{\frac34}(2\sqrt{s})\big)}
{t\big(- J_{-\frac14}(2\sqrt{s}) +{\kappa} J_{\frac14}(2\sqrt{s})\big)} \label{ls} \\
& =\frac{-i\bigl(8{\kappa}\varepsilon^{\frac12}M_{-\frac34}({s}) + 
\varepsilon^2 t^3M_{\frac34}({s})\bigr)}
{2\bigl(- 2M_{-\frac14}(s)+{\kappa}\varepsilon^\frac12 t M_{\frac14}({s})\bigr)}\,.
\label{lt}
\end{align} 
{\rm (3)} For the solution $l_{\varepsilon}(t,\kappa)$ we have as $t \to 0$
\begin{equation} 
l_{\varepsilon}(t,\kappa) = i a\varepsilon^{\frac12}
\bigl( 1 + a\varepsilon^{\frac12}t+
\bigl(a\varepsilon^\frac{1}{2}t\bigr)^2+ O\big(\varepsilon^{\frac12}t\big)^3\bigr), 
\quad a=2\kappa \Gamma(3/4)/\Gamma(1/4)\,.
\label{223-at0}
\end{equation} 
\end{lemma} 
\begin{proof} Define 
$w_1(s)= s^{1/8}\big( aJ_\frac14 (2\sqrt{s} )+ 
bY_{\frac14}(2\sqrt{s})\big)$ for $s\geq 0$. Davis\cite{D}, pages 67-78, shows 
that general solution of \eqref{Riccati-eq} is given by 
\begin{equation} \label{w-tilde}
l_{\varepsilon}(t)= \frac{-4isw_1'(s)}{tw_1(s)}
\end{equation} 
with arbitrary constants $a$ and $b$ which are not $\mathbb{R}$-related. 
If we use $Y_{\frac14}=J_{\frac14}-\sqrt{2}J_{-\frac14}$ and set  
$\kappa= (a+b)/(\sqrt{2}b)\in (\mathbb{C}\cup\{\infty\})
\setminus\mathbb{R}$, the right hand side of \eqref{w-tilde} 
becomes $l_{\varepsilon}(t,\kappa)$ of \eqref{w}. Lemma~\ref{hoho} implies part (1). 

To prove part (2) we use the recurrence formula of Bessel functions 
(see (10.6.5) in DLMF\cite{DLMF}):
\begin{lemma} Let $\mathcal{C}_\nu(z)$ be any of $J_\nu(z), Y_\nu(z), H_{\nu}^{(1)}(z), 
H_{\nu}^{(2)}(z)$ or any nontrivial linear combination of these functions, 
the coefficients in which are independent of $z$ and $\nu$. 
Define 
$f_\nu(z)=z^p \mathcal{C}_\nu(\lambda{z}^q)$, where $p,q$, and $\lambda\not=0$ 
are real or complex constants, then
\begin{equation} \label{Bessel}
zf_\nu'(z)=\lambda q z^q f_{\nu-1}(z)+(p-{\nu}q)f_\nu(z)
\end{equation} 
as long as the principal branch is considered for $\mathcal{C}_\nu(z)$, 
\end{lemma} 

\noindent 
Conside $f_\nu(z)$ for $\mathcal{C}_\nu(z)=aJ_\nu(z)+ bY_\nu(z)$ with 
\[
p=\frac18, \quad q=\frac12, \quad \nu=\frac14, \quad \lambda=2.
\]
Then $\lambda{q}=1, \ p-{\nu}q=0$ and \eqref{Bessel} implies that 
for $w_1(s)$ of \eqref{w-tilde} we have   
\begin{equation} \label{gather-1}
l_\varepsilon(t)= 
\frac{-4i}{t} \frac{s\big(s^{\frac18}\mathcal{C}_{\frac14}(2\sqrt{s})\big)'}
{s^\frac18 \mathcal{C}_{\frac14}(2\sqrt{s})}\, 
=\frac{-4i}{t}\frac{s^{\frac58}\mathcal{C}_{-\frac34}(2\sqrt{s})}
{s^\frac18 \mathcal{C}_{\frac14}(2\sqrt{s})} \, .
\end{equation} 
In the right hand side of \eqref{gather-1} substitute 
\begin{gather*}
\mathcal{C}_{\frac14}(z)= a J_{\frac14}(z) + b Y_{\frac14}(z)= 
(a+b)J_{\frac14}(z)- \sqrt{2} b J_{-\frac14}(z),\\
\mathcal{C}_{-\frac34}(z)= a J_{-\frac34}(z) + b Y_{-\frac34}(z)= 
(a+b)J_{-\frac34}(z) + \sqrt{2}b J_{\frac34}(z)
\end{gather*}
with $z=2\sqrt{s}$ and reduce by the common factor $\sqrt{2}b$. 
In the denominator we have  
\begin{equation} \label{numerator}
\frac{s^\frac18}{\sqrt{2}b} \mathcal{C}_{\frac14}(2\sqrt{s})
= s^{1/8}
\big(- J_{-\frac14}(2\sqrt{s}) +{\kappa} J_{\frac14}(2\sqrt{s})\big)= w(s,\kappa)
\end{equation} 
and in the numerator  
\begin{equation} \label{denomi}
\frac{s^{\frac58}}{\sqrt{2}b}\mathcal{C}_{-\frac34}(2\sqrt{s})
= s^{\frac58}\big({\kappa}J_{-\frac34}(z) + J_{\frac34}(z)\big)
= \frac{\sqrt{\varepsilon}\kappa t}{2} M_{-\frac34}(s) + s M_{\frac34}(s). 
\end{equation} 
Plugging these in \eqref{gather-1} we obtain \eqref{ls}. 
If we use \eqref{M1/4} and the last expression in \eqref{denomi} 
we obtain \eqref{lt} which manifests that $l_\varepsilon(t,\kappa)$ 
is a meromorphic function of $t$. 

To prove part (3), we use \eqref{lt}. The numerator has an asymptotic expansion 
\begin{equation*}
-i(8\kappa \varepsilon^{\frac12}\Gamma(1/4)^{-1}+ O(\varepsilon^2 t^3))
\end{equation*}
 and the denominator  has an asymptotic expansion
\begin{equation*}
2(-2\Gamma(3/4)^{-1} + \kappa \varepsilon^{\frac12} t \Gamma(5/4)^{-1}+ O(\varepsilon^2 t^4))
\end{equation*} 
as $t \to 0$, 
hence 
\begin{align*}
l_\varepsilon(t,\kappa)& = \frac{-i8\kappa \varepsilon^{\frac12}\Gamma(1/4)^{-1}+ O(\varepsilon^2 t^3)}
{2(-2\Gamma(3/4)^{-1} + \kappa \varepsilon^{\frac12} t \Gamma(5/4)^{-1}+ O(\varepsilon^2 t^4))} \\
&= \Bigl(\frac{2i\kappa\varepsilon^{\frac12}\Gamma(3/4)}{\Gamma(1/4)}+ O(\varepsilon^2 t^3)\Bigr)
\Bigl(1-\frac{2\kappa\varepsilon^{\frac12} t \Gamma(3/4)}{\Gamma(1/4)}+ O(\varepsilon^2 t^4)\Bigr)^{-1} \\
& = \frac{ia \varepsilon^{\frac12}}
{1-a\varepsilon^{\frac12}t + O(\varepsilon^2 t^4)}+ O({\varepsilon}^{2}t^3), \quad 
a= \frac{2\kappa \Gamma(3/4)}{\Gamma(1/4)}.
\end{align*}
Statement (3) follows. 
\end{proof} 

\noindent
\textbf{The initial condition.}
Having obtained the general solution $l_{\varepsilon}(t,\kappa)$ of 
\eqref{Riccati-eq}, we need to determine $\kappa=\kappa_\varepsilon$ such that 
the initial condition $l_{\varepsilon}(-1/\varepsilon,\kappa_\varepsilon)= 1$ 
is satisfied. We define 
\[
l_{\varepsilon}^\ast(t)= l_{\varepsilon}(t,\kappa_\varepsilon) \quad 
\mbox{and} \quad 
\tilde{l}_{\varepsilon}(t)=-l_{\varepsilon}^\ast(-t).
\]
We have introduced $\tilde{l}_{\varepsilon}(t)$ as we want to deal with 
a positive variable. Then \eqref{lt} implies for $t>0$ that  
\begin{align}
\tilde{l}_{\varepsilon}(t)
&= 
\frac{-i\big(8{(-\kappa_\varepsilon)}\varepsilon^{\frac12}M_{-\frac34}({s}) + 
\varepsilon^2 t^3M_{\frac34}({s})\big)}
{2\big(- 2M_{-\frac14}(s)+{(-\kappa_\varepsilon)}\varepsilon^\frac12 t M_{\frac14}({s})\big)} 
\notag \\
&=\frac{-4i s^\frac12
\big({-\kappa_\varepsilon}J_{-\frac34}(2\sqrt{s}) + J_{\frac34}(2\sqrt{s})\big)}
{t(- J_{-\frac14}(2\sqrt{s}) -{\kappa_\varepsilon} J_{\frac14}(2\sqrt{s}))} 
\notag \\%
& = l_{\varepsilon}(t,-\kappa_\varepsilon) 
={ - i \frac{\tilde{w}'_{\varepsilon}(t, -\kappa_\varepsilon )}
{\tilde{w}_{\varepsilon}(t, -\kappa_\varepsilon)}\,.}
\label{tldast}
\end{align} 
Thus, $\tilde{l}_{\varepsilon}(1/\varepsilon)=-1$ is satisfied if (and only if)  
\[
\left. -1 = 
\frac{-4i s^\frac12
\big({-\kappa_\varepsilon}J_{-\frac34}(2\sqrt{s}) + J_{\frac34}(2\sqrt{s})\big)}
{t(- J_{-\frac14}(2\sqrt{s}) -{\kappa_\varepsilon} J_{\frac14}(2\sqrt{s}))}\right\vert_{t=1/\varepsilon} 
=
\frac{i\big(-{\kappa_\varepsilon}J_{-\frac34}(1/2\varepsilon) + J_{\frac34}(1/2\varepsilon)\big)}
{J_{-\frac14}(1/2\varepsilon) +{\kappa_\varepsilon} J_{\frac14}(1/2\varepsilon)}
\]
where we used $2\sqrt{s}=1/(2\varepsilon)$ when $t=1/\varepsilon$ in the  last 
expression. Solving this equation for $\kappa_\varepsilon$ leads to 
\begin{equation} 
\kappa_\varepsilon = - \left.
\frac{J_{-\frac14}+ i J_{\frac34}}
{J_{\frac14}-i J_{-\frac34}}\right|_{\frac1{2\varepsilon}}
\label{asymp}
\end{equation} 

We recall the following special case of (10.17.3) in DLMF\cite{DLMF}.
\begin{lemma} 
Assume $x$ real and let $\omega=x-\frac12 \nu\pi- \frac{\pi}{4}$ . Then as $x \to \infty$
\begin{equation} \label{Bessel-asymptotic}
J_\nu(x)= \sqrt{\frac{2}{\pi{x}}}
\Bigl(
\cos\, \omega -  
\frac{(4\nu^2-1)}{8x}\sin\, \omega  + 
O\bigl(\frac1{x^2}\bigr)
\Bigr)\,.
\end{equation} 	
\end{lemma} 
Application of this result to the right hand side of \eqref{asymp} yields  
\begin{equation} \label{gamma} 
\kappa_\varepsilon = - e^{\frac{i\pi}{4}}+ (2\sqrt{2})^{-1} \varepsilon
e^{-i\left(\frac{1}{\varepsilon}+\frac{\pi}{4}\right)}+ 
O(\varepsilon^2)= - e^{\frac{i\pi}{4}} + O(\varepsilon), \quad \varepsilon \to 0.
\end{equation} 
We omit the details.
 
\begin{lemma} \label{negative-t} 
The solution ${l}^\ast_{\varepsilon}(t)$ of the initial value problem for the 
Riccati equation 
\begin{equation} 
{l}^{\ast{'}}_{\varepsilon}(t) + i {l}^\ast_{\varepsilon}(t)^2 = i\varepsilon^2 t^2 , \quad 
{l}^\ast_{\varepsilon}(-1/\varepsilon)=  1 
\end{equation} 
is given by \eqref{ls} with $\kappa$ given by $\kappa_\varepsilon$ of \eqref{asymp}: 
\begin{equation} \label{tlt}
{l}_{\varepsilon}^\ast (t)= \frac{-4i s^\frac12
\big({\kappa_\varepsilon}J_{-\frac34}(2\sqrt{s}) + J_{\frac34}(2\sqrt{s})\big)}
{t(- J_{-\frac14}(2\sqrt{s}) +{\kappa_\varepsilon} J_{\frac14}(2\sqrt{s}))}, \quad 
s= \frac{\varepsilon^2 t^4}{16},
\end{equation}
where the principal branch is assumed for Bessel functions. 
\end{lemma} 

\noindent
\textbf{Asymptotic behavior of ${l}^\ast_{\varepsilon}(1/\varepsilon)$ as $\varepsilon \to 0$.}  
\begin{lemma} As $\varepsilon \to 0$, we have 
\begin{equation} \label{llep}
l^\ast_\varepsilon (1/\varepsilon)= \frac{1-2\sqrt{2}i\cos(1/\varepsilon)}{3+2\sqrt{2}\sin(1/\varepsilon)} 
+ O(\varepsilon)
\end{equation} 
and $\real l^\ast_\varepsilon(1/\varepsilon)$ oscillates between $(3+2\sqrt{2})^{-1}$ and 
$(3-2\sqrt{2})^{-1}$ as $\varepsilon \to 0$. 
\end{lemma} 
\begin{proof} 
From \eqref{tlt} we have  
\begin{equation} \label{323}
\left.
l^\ast_\varepsilon ({1}/{\varepsilon})= 
\frac{-i \big({\kappa_\varepsilon}J_{-\frac34} + J_{\frac34}\big)}
{-J_{-\frac14} +{\kappa_\varepsilon} J_{\frac14}}
\right|_{\frac1{2\varepsilon}}, \quad 
\kappa_\varepsilon = - \left. \frac{J_{-\frac14}+ i J_{\frac34}}
{J_{\frac14}-i J_{-\frac34}}\right|_{\frac{1}{2\varepsilon}}.  
\end{equation} 
Thus, 
\begin{equation} \label{asymp-3}
\left. l^\ast_\varepsilon ({1}/{\varepsilon})= 
\frac{2J_{\frac34}J_{-\frac34}+ 
i\bigl(J_{\frac34}J_{\frac14}-J_{-\frac14}J_{-\frac34}\bigr)}
{2J_{-\frac14}J_{\frac14}+ 
i\bigl(J_{\frac14}J_{\frac34} -J_{-\frac14}J_{-\frac34}\bigr)
}\right|_{ \frac1{2\varepsilon}}
\end{equation} 
and we may compute the asymptotic value of \eqref{asymp-3} 
as $\varepsilon \to 0$ by applying once more \eqref{Bessel-asymptotic}. 
This yields \eqref{llep} and the 
lemma follows. 
\end{proof} 

\noindent
\textbf{The amplitude function $m_{\varepsilon}(t)$.} 
We next solve initial value problem \eqref{m} associated with 
$l^\ast_{\varepsilon}(t)$ which reads 
\[
\frac{{m}'_{\varepsilon}(t)}{{m}_{\varepsilon}(t)}= \frac{l_{\varepsilon}^\ast (t)}{2i}, \quad 
m_\varepsilon(-1/{\varepsilon})= \pi^{-1/4}.
\] 
For the same reason as before, we consider 
$\tilde{m}_{\varepsilon}(t)= m_{\varepsilon}(-t)$. The expression \eqref{tldast} 
for $\tilde{l}_\varepsilon(t)$ implies  
\begin{equation} \label{tmderi}
\frac{\tilde{m}_{\varepsilon}'(t)}{\tilde{m}_{\varepsilon}(t)}= 
- \frac{{m}_{\varepsilon}'(-t)}{{m}_{\varepsilon}(-t)}
= -\frac{l^\ast_\varepsilon(-t)}{2i}
= \frac{\tilde{l}_\varepsilon(t)}{2i}
= - \frac{\tilde{w}_{\varepsilon}'(t,-\kappa_\varepsilon)}{2 \tilde{w}_{\varepsilon}(t,-\kappa_\varepsilon)}.
\end{equation} 
Recall \eqref{w-def}  and \eqref{winM} for the definition of 
$\tilde{w}_{\varepsilon}(t,\kappa)$.
Integrating \eqref{tmderi} yields 
$\tilde{m}_\varepsilon(t)= A_\varepsilon \tilde{w}_{\varepsilon}(t,-\kappa_\varepsilon)^{-1/2}$ for a constant 
$A_\varepsilon$ for $t>0$, viz.  
\begin{align} \label{m-sol-negative} 
{m}_{\varepsilon}(- t) & 
= A_\varepsilon \big(-s^{1/8}J_{-\frac14}(2\sqrt{s}) -{\kappa_\varepsilon}s^{1/8}
J_{\frac14}(2\sqrt{s})\big)^{-1/2} \\
&  = A_\varepsilon 
\big(- M_{-\frac14}({s}) -{\kappa_\varepsilon} t M_{\frac14}({s})\big)^{-1/2}. 
\label{m-sol-negative-1}
\end{align} 
Thus the initial condition 
${m}_{\varepsilon}(- 1/\varepsilon)= \pi^{-\frac14}$ is satisfied if  
\begin{equation}\label{nm}
\pi^{-\frac14}=A_\varepsilon \big(-s^{1/8}J_{-\frac14}(2\sqrt{s}) -{\kappa_\varepsilon}s^{1/8}
J_{\frac14}(2\sqrt{s})\big)^{-1/2}\vert_{t=1/\varepsilon}. 
\end{equation} 
By virtue of 
\eqref{Bessel-asymptotic} and \eqref{gamma}, 
$(\cdots)$ on the right hand side is equal to (with $\alpha=\frac1{2\varepsilon}-\frac{\pi}4$)
\[ \frac{2^{\frac12}\varepsilon^{\frac14}}
{\pi^{\frac12}}
\left(
-\cos\left(\alpha+\frac{\pi}{8}\right)
+ e^{\frac{i\pi}{4}}
\cos\left(\alpha-\frac{\pi}{8}\right) +O(\varepsilon)\right) 
= \frac{\varepsilon^{\frac14}}
{\pi^{\frac12}}e^{-i\left(\frac{1}{2\varepsilon}- \frac{7\pi}{8}\right)}+O(\varepsilon^\frac54),
\]
and  we have 
\begin{equation} \label{Aep}
A_\varepsilon = \frac{\varepsilon^{\frac18}}{\pi^\frac12}
e^{-i\left(\frac{1}{4\varepsilon}- \frac{7\pi}{16}\right)}(1+ O(\varepsilon)).
\end{equation} 
\eqref{m-sol-negative-1} implies that $m_\varepsilon(t)$ is given by 
changing $\kappa_\varepsilon$ to $-\kappa_{\varepsilon}$ in the right hand side of  \eqref{m-sol-negative} 
or \eqref{m-sol-negative-1}. This proves the first statement of 
the following lemma.

\begin{lemma} {\rm (1)} 
The solution of the initial value problem \eqref{m} associated with 
$l^\ast_{\varepsilon}(t)$ is given by 
\begin{equation} \label{sol-m}
m_\varepsilon(t) = A_\varepsilon \big(-s^{1/8}J_{-\frac14}(2\sqrt{s}) + {\kappa_\varepsilon}s^{1/8}
J_{\frac14}(2\sqrt{s})\big)^{-1/2},
\end{equation} 
where $\kappa_\varepsilon$ and $A_\varepsilon$ are asymptotically given by \eqref{gamma} 
and \eqref{Aep} respectively and the branch of the square root should be 
chosen such that $m_\varepsilon(-1/\varepsilon)= \pi^{-\frac14}$. 

\noindent
{\rm (2)} As $\varepsilon \to 0$, 
\begin{equation} \label{m-asymptotic}
m_\varepsilon(1/{\varepsilon})= 
\frac{\pi^{-1/4}}{(\sqrt{2}e^{i/\varepsilon}+ i)^{1/2}} + O(\varepsilon),
\end{equation} 
where the branch of the square root should be chosen by the continuity.
\end{lemma}
\begin{proof} By virtue of \eqref{nm} and \eqref{sol-m}, 
\[
m_\varepsilon(1/{\varepsilon})^2 = \pi^{-\frac12}
\frac{-J_{-\frac14} - {\kappa_\varepsilon}J_{\frac14}}
{-J_{-\frac14} + {\kappa_\varepsilon}J_{\frac14}}\Big\vert_{\frac1{2\varepsilon}}
\]
and we compute the asymptotic value of the right side by using 
\eqref{Bessel-asymptotic}. We obtain \eqref{m-asymptotic}. 
\end{proof} 

\noindent
\textbf{Asymptotic behavior at $t=0$.} 
In the following section we need $l^\ast_{\varepsilon}(0)$ and $m_\varepsilon(0)$. 
We already computed $l^\ast_{\varepsilon}(0)=ia\varepsilon^{\frac12}$ in \eqref{223-at0} where 
$\kappa$ in the expression for $a$ should be taken as $\kappa=\kappa_\varepsilon$ (see \eqref{gamma}).  The next lemma immediately follows from 
\eqref{m-sol-negative-1} or \eqref{sol-m}. 

\begin{lemma} \label{at0} 
As $t\to 0$,  $m_\varepsilon(t)$ has the following 
asymptotic expansion, uniformly for $0<\varepsilon<1$,
\begin{equation}  
m_\varepsilon(t) =-i A_\varepsilon\Gamma(3/4)^{1/2}
\bigl(1 +\frac{2\Gamma(3/4)}{\Gamma(1/4)}\kappa_\varepsilon t + O(t^2) \bigr),  
\label{m223}
\end{equation} 
where $A_\varepsilon$ and $\kappa_\varepsilon$ are as in \eqref{Aep} and \eqref{gamma},
respectively.  
\end{lemma} 

Lemma~\ref{at0}
shows how the adiabatic approximation breaks down as $t \to 0$: 
The adiabatic approximation would yield $l_{\varepsilon}(t)=\varepsilon {t}/2$ for (minus) the 
exponent of the Gaussian as $\varepsilon \to 0$ whereas the leading term in 
\eqref{223-at0} is $i a\varepsilon^{\frac12}$ which does not go to zero as $t\to0$.  
The corresponding term of order $\varepsilon {t}$ appears 
only as the second term $i a^2 \varepsilon {t}\sim  C^2 \varepsilon {t}/2$, 
$C=2\Gamma(3/4)/\Gamma(1/4)\approx 0.676$. The state at time $t=0$, $v_\varepsilon(0,x)$, 
is a Gaussian, which is a result of general theorems (see Hagedorn et al.\cite{HLS}),  but 
the speed of spreading is $C\varepsilon^{1/4}\sqrt{t}$ times slower 
than the one given by the adiabatic approximation and, at time zero, it 
remains as a finite Gaussian of size $C\varepsilon^{-1/4}$ whereas the adiabatic 
approximation gives a completely flat Gaussian.

\noindent
\textbf{Behavior of $v_\varepsilon(1/\varepsilon)$ as $\varepsilon \to 0$ 
and the survival probability.} 
The following theorem states the main result of this section for the case $L=0$. 
The theorem explicitly exhibits that the state at the microscopic time $1/\varepsilon$, 
when the Hamiltonian returns to the initial $-(1/2)d^2/dx^2 + (1/2)x^2$,
is highly oscillating as $\varepsilon \to 0$ and the 
adiabatic approximation is completely broken down.  

We introduce notation for the leading terms in the asymptotic expansions \eqref{llep} and \eqref{m-asymptotic}. We define
\begin{align}
l^\ast_{\varepsilon,0}(1/\varepsilon)&=\frac{1-2\sqrt{2}i\cos(1/\varepsilon)}{3+2\sqrt{2}\sin(1/\varepsilon)} ,\label{l0}\\
m_{\varepsilon,0}(1/\varepsilon)&=\frac{\pi^{-1/4}}{(\sqrt{2}e^{i/\varepsilon}+ i)^{1/2}}.\label{m0}
\end{align}

We have proven the following theorem:

\begin{theorem} \label{1} 
{\rm (1)} Let $l^\ast_{\varepsilon,0}(1/\varepsilon)$ and $m_{\varepsilon,0}(1/\varepsilon)$ be given by  
\eqref{l0} and \eqref{m0}, respectively. Then 
the solution $v_\varepsilon(t,x)$ of the 
initial value problem \eqref{harmonic} satisfies as $\varepsilon \to 0$,  
\begin{equation} \label{wave-asymp}
\|  
v_\varepsilon(1/{\varepsilon},x) 
- m_{\varepsilon,0}(1/\varepsilon)e^{-l^\ast_{\varepsilon,0} (1/\varepsilon)x^2/2}\| \leq C \varepsilon.
\end{equation} 
We have 
\begin{equation} 
|m_{\varepsilon,0}(1/\varepsilon)|^4= \pi^{-1}
\bigl(3+2\sqrt{2}\sin(1/{\varepsilon})\bigr)^{-1}
= \pi^{-1} \real l^\ast_{\varepsilon,0}(1/\varepsilon). 
\label{mlep}
\end{equation} 
{\rm (2)} The survival probability  
of the ground state $\varphi_0(x)= \pi^{-\frac14}e^{-x^2/2}$ at time $1/\varepsilon$ is 
equal to $1/\sqrt{2}+ O(\varepsilon)$.
\end{theorem}
\begin{remark}
The survival probability in part (2) was also computed in Theorem 1 of Bachmann et al.\cite{BFG}.
\end{remark}
\begin{proof} 
Since $\real\ell^{\ast}_{\varepsilon}
(1/\varepsilon)\geq(3+2\sqrt{2})^{-1}$, part (1) is obvious.
We only prove (2). Using \eqref{wave-asymp} and explicitly computing 
the Gaussian integral, we obtain  
\begin{align*}  
\langle v_\varepsilon(1/\varepsilon,x), \varphi_0(x)\rangle
& = \int_{\mathbb{R}} \pi^{-\frac14} e^{-x^2/2}
m_{\varepsilon,0}(1/\varepsilon)
e^{-l^\ast_{\varepsilon,0}(1/\varepsilon)x^2/2} dx + O(\varepsilon) \\ 
& = 
\frac{\sqrt{2}\pi^{1/4}m_{\varepsilon,0}(1/\varepsilon)}
{(1+l^\ast_{\varepsilon,0}(1/\varepsilon))^{1/2}}  + O(\varepsilon).
\end{align*}
Insert the expressions from \eqref{l0} and \eqref{m0}.
Since 
\begin{equation*}
(\sqrt{2}e^{i/\varepsilon}+ i)
\Big(1+\frac{1-2\sqrt{2}i\cos(1/\varepsilon)}{3+2\sqrt{2}\sin(1/\varepsilon)}\Big) 
= { 
e^{i/\varepsilon}\frac{6\sqrt{2}+ 8\sin(1/\varepsilon)}{3+2\sqrt{2}\sin(1/\varepsilon)} = 
e^{i/\varepsilon} 2\sqrt{2},}
\end{equation*}
we conclude that 
\begin{equation*}
|\langle v_\varepsilon(1/\varepsilon,x), \varphi_0(x)\rangle|^2 
= { 
\frac{2(3+2\sqrt{2}\sin(1/\varepsilon))}{6\sqrt{2}+ 8\sin(1/\varepsilon)} + O(\varepsilon)}
= \frac{1}{\sqrt{2}} + O(\varepsilon).
\end{equation*} 
\end{proof} 

\section{The case $L >0$.} 
We next study the case $L>0$ and examine how the asymptotic behavior
as $\varepsilon \to 0$ of the solution depends on the macroscopic length $L$ of time 
which the particle has spent in the continuum of $-(1/2)\partial_x^2$. 
We let $v_\varepsilon(t,x)$ 
be the solution of the initial value problem \eqref{harmonic} 
with $L>0$. Then, by translating in time the result for the case 
$L=0$ by $-L/\varepsilon$,   
we see from \eqref{223-at0} with $\kappa=\kappa_\varepsilon$ and \eqref{m223} that 
\[
v_\varepsilon(-{L}/{\varepsilon}, x) = -i A_\varepsilon \, \Gamma(3/4)^{1/2}  
e^{-ia_\varepsilon \varepsilon^{1/2}x^2/2}, \ \  
a_\varepsilon = \frac{2\Gamma(3/4)}{\Gamma(1/4)}\kappa_\varepsilon\, .
\]

\noindent
\textbf{Solution at the time exiting the continuum.} 
We may explicitly compute 
\begin{align}
v_\varepsilon({L}/{\varepsilon}, x)& = \bigl(e^{-2iLH_0/\varepsilon}
v_\varepsilon(-{L}/{\varepsilon})\bigr)(x) \notag \\
& = \frac{-i A_\varepsilon \, \Gamma(3/4)^{1/2}
e^{-\frac{\pi{i}}{4}}}
{(4{\pi}L/\varepsilon)^{1/2}} 
\int_{\mathbb{R}} e^{{ \frac{i\varepsilon (x-y)^2}{4L}}
-i\frac{a_\varepsilon  \varepsilon^{1/2}y^2}{2}}dy  \notag \\
&  = 
\frac{-A_\varepsilon \Gamma(3/4)^{\frac12}{\varepsilon}^{\frac14}}
{{ (-2a_\varepsilon L+\sqrt{\varepsilon})^{1/2}}} 
e^{\frac{-i\varepsilon{a_\varepsilon }}{{ 2(-2a_\varepsilon L+\sqrt{\varepsilon})}}x^2}\,. \label{334}
\end{align}

\noindent
\textbf{Solution after the particle exits the continuum.}
We want to evaluate at time $t= \frac{L+1}{\varepsilon}$ the solution of 
\[
i\partial_t v_\varepsilon(t,x) 
= -\frac12\partial_x^2 v_\varepsilon + \frac{(t-L/\varepsilon)^2 \varepsilon^2 x^2}{2}v_\varepsilon
\]
when $v_\varepsilon(L/\varepsilon,x)$ is given by \eqref{334}.  
Translation of $t$ by $L/\varepsilon$ once again shows that 
$v_\varepsilon((L+1)/\varepsilon,x)=z_\varepsilon(1/\varepsilon,x)$, where 
$z_\varepsilon(t,x)$ is the solution of   
\begin{align}
& i\partial_t z_\varepsilon(t,x) 
= -\frac12\partial_x^2 z_\varepsilon + \frac{t^2 \varepsilon^2 x^2}{2} z_\varepsilon, \\ 
& z_\varepsilon(0,x)= 
\frac{-A_\varepsilon \Gamma(3/4)^{\frac12}{\varepsilon}^{\frac14}}
{{ (-2a_\varepsilon L+\sqrt{\varepsilon})^{1/2}}} 
e^{\frac{-i\varepsilon{a_\varepsilon }}{{ 2(-2a_\varepsilon L+\sqrt{\varepsilon})}}x^2}\,. 
   \label{init-L}
\end{align}
We know from \eqref{ls} that $z_\varepsilon (t,x)$ is of the form 
\[
z_\varepsilon(t,x)= {m_{\varepsilon}^\ast(t)}{e^{-l_\varepsilon(t,\gamma)x^2/2}}
\] 
where $l_\varepsilon(t,\gamma)$ and $m_\varepsilon^\ast(t)$ are  
given by \eqref{ls} and \eqref{sol-m} respectively, with $\gamma$ 
in place of $\kappa$ and $B_\varepsilon$ in place of $A_\varepsilon$, in particular, 
\begin{gather}
l_{\varepsilon}(t,\gamma) =\frac{-4is^\frac12
\big({\gamma}J_{-\frac34}(2\sqrt{s}) + J_{\frac34}(2\sqrt{s})\big)}
{t\big(- J_{-\frac14}(2\sqrt{s}) +\gamma J_{\frac14}(2\sqrt{s})\big)} ,
\label{ls-c} \\
m_\varepsilon^\ast (t) = B_\varepsilon \big(-s^{1/8}J_{-\frac14}(2\sqrt{s}) + {\gamma}s^{1/8}
J_{\frac14}(2\sqrt{s})\big)^{-1/2} . \label{sol-m-after}
\end{gather}
We will choose $\gamma$ and $B_\varepsilon$ such that the initial condition 
\eqref{init-L} is met, viz. 
\begin{gather}
l_\varepsilon(0,\gamma) = 
\frac{i\varepsilon{a_\varepsilon}}{-2{a_\varepsilon}L+\sqrt{\varepsilon}}, \label{l10-0}\\
m_\varepsilon^\ast (0) = 
\frac{-A_\varepsilon \Gamma(3/4)^{\frac12} {\varepsilon}^{\frac14}}
{(-2{a_\varepsilon}L+\sqrt{\varepsilon})^{1/2}}. \label{m0ep-0}
\end{gather}
By virtue of \eqref{Bessel-expansion} we may evaluate $l_\varepsilon(0,\gamma)$ of 
\eqref{ls-c} and $m_\varepsilon^\ast(0)$ of \eqref{sol-m-after}:   
\begin{equation} \label{l10}
l_\varepsilon(0,\gamma)= \frac{2i\gamma \varepsilon^{\frac12}\Gamma(3/4)}{\Gamma(1/4)}, \quad 
m_\varepsilon^\ast (0)= -i B_\varepsilon \Gamma(3/4)^{1/2}.  
\end{equation} 
Equating the right hand sides of \eqref{l10-0} and \eqref{m0ep-0} with those of 
\eqref{l10}, we have 
\begin{equation} \label{CC}
 \gamma=\frac{\varepsilon^\frac12 \kappa_\varepsilon}
{ -2a_\varepsilon L+ \sqrt{\varepsilon}}, \quad 
B_\varepsilon = \frac{-iA_{\varepsilon}\varepsilon^{\frac14}}{
{ (-2{a_\varepsilon}L +\sqrt{\varepsilon})^{1/2}}
}.
\end{equation} 
Hereafter we write $C_1= \Gamma(3/4)/\Gamma(1/4)$ so that 
$a_\varepsilon= 2\kappa_\varepsilon C_1$. Note that $\gamma=\kappa_\varepsilon$ and $B_\varepsilon= A_\varepsilon$ when $L=0$ as 
they should be. 

\noindent
\textbf{Solution when the Hamiltonian returns to $-(1/2)d^2/dx^2 + x^2/2$.}
We study the behavior as $\varepsilon \to 0$ of 
$l_\varepsilon(1/\varepsilon, \gamma)$ and $m_\varepsilon(1/\varepsilon)$. They are given by 
\begin{gather} \label{l-outgoing}
l_\varepsilon ({1}/{\varepsilon}, \gamma)= \left.
\frac{-i\bigl(\gamma J_{-\frac34}(z) + 
J_{\frac34}(z)\bigr)}
{\bigl(-J_{-\frac14}(z) + \gamma J_{\frac14}(z)\bigr)}
\right\vert_{z=\frac1{2\varepsilon}}, \\
m_\varepsilon(1/\varepsilon) = B_\varepsilon \bigl(-(z/2)^{1/4}J_{-\frac14}(z) + {\gamma}(z/2)^{1/4}
J_{\frac14}(z)\bigr)^{-1/2} \Big\vert_{z=\frac1{2\varepsilon}}. 
\label{m-outgoing}
\end{gather}
We substitute the first of \eqref{CC} for $\gamma$, 
which yields 
\begin{equation} 
\left.
l_\varepsilon ({1}/{\varepsilon}, \gamma)=
i \frac{{ - 4\kappa_{\varepsilon} C_1 L} J_{\frac34}(z)+ \varepsilon^{1/2}
\big(\kappa_{\varepsilon}J_{-\frac34}(z)
+J_{\frac34}(z)\big)}
{{- 4\kappa_{\varepsilon}C_1 L} J_{-\frac14}(z)- \varepsilon^{1/2}\big(
\kappa_{\varepsilon}J_{\frac14}(z)- J_{-\frac14}(z)\big)}
\right\vert_{z=\frac1{2\varepsilon}}\,.
\label{el1}
\end{equation} 
We substitute $\kappa_{\varepsilon}= - e^{\frac{i\pi}{4}} + O(\varepsilon)$ in \eqref{el1} 
and use \eqref{Bessel-asymptotic}. Denote 
\[
L_1= -4C_1 L, \quad 
\alpha=(2\varepsilon)^{-1}-4^{-1}\pi.
\]
Then, as $\varepsilon \to 0$, $l_\varepsilon(1/{\varepsilon})$ (omitting $\gamma$ in the notation) is asymptotically 
equal to 
\begin{align*}  
& 
i\; \frac{
{ -L_1} e^{\frac{i\pi}{4}}
\cos(\alpha -\frac{3\pi}{8}) + 
\varepsilon^{1/2}\bigl(-e^{\frac{i\pi}{4}} 
\cos(\alpha+\frac{3\pi}{8} ) 
+\cos(\alpha-\frac{3\pi}{8} )
\bigr)+O(\varepsilon)}
{{ -L_1} e^{\frac{i\pi}{4}}\cos(\alpha+\frac{\pi}{8})+ 
\varepsilon^{1/2}\bigl( 
\cos(\alpha+\frac{\pi}{8})+  e^{\frac{i\pi}{4}}
\cos(\alpha-\frac{\pi}{8})\bigr)
+O(\varepsilon) }  \notag \\
& = 
i\; \frac{{ -L_1} 
\sin(\frac1{2\varepsilon} -\frac{\pi}{8}) + 
\varepsilon^{1/2}\bigl(e^{\frac{-i\pi}{4}} \sin(\frac1{2\varepsilon}-
\frac{\pi}{8} ) 
+ \sin(\frac1{2\varepsilon}-\frac{3\pi}{8} )
\bigr)+O(\varepsilon)}
 {-L_1 
\cos(\frac1{2\varepsilon}-\frac{\pi}{8})+ 
\varepsilon^{1/2}\big(e^{\frac{-i\pi}{4}} 
\cos(\frac1{2\varepsilon}- \frac{\pi}{8}) 
+ \cos(\frac1{2\varepsilon}- \frac{3\pi}{8})  
\bigr)
+O(\varepsilon) }\,  .
\end{align*} 

After a simple but tedious computation we simplify the equation above 
and obtain the following lemma. 

\begin{lemma} Define $\rho= \frac1{2\varepsilon}-\frac{\pi}{8}$ and 
${ B=L_1 - \sqrt{2\varepsilon}}$. 
As $\varepsilon \to 0$, we have 
\begin{equation} 
l_\varepsilon(1/{\varepsilon})= 
\frac{\varepsilon + i\left(B^2 \sin 2\rho + 
\sqrt{2\varepsilon}{ B}\cos 2\rho \right)+O(\varepsilon)}
{B^2+\varepsilon + B^2\cos2\rho -\sqrt{2\varepsilon}B \sin 2\rho + O(\varepsilon)}
\end{equation} 
\end{lemma} 
We note that in the denominator we have
\begin{align}
A(\varepsilon)& \equiv B^2+\varepsilon + B^2\cos2\rho 
-\sqrt{2\varepsilon}B \sin 2\rho  \notag \\
& = B^2 \bigl(
1 +\frac{\varepsilon}{B^2} + 
\bigl(1+ \frac{2\varepsilon}{B^2}\bigr)^{\frac12}\cos(2\rho + \beta)
\bigr) \geq C\varepsilon^2 
\end{align} 
where 
\[
\sin{\, \beta} = \frac{\sqrt{2\varepsilon}}{(B^2+ 2\varepsilon)^{\frac12}}.
\] 
However, $A(\varepsilon)+O(\varepsilon)$ can be controlled only when 
$A(\varepsilon)\geq C \varepsilon^{1-\delta}$ for a 
$\delta>0$ and this requires 
\begin{equation} \label{sufficient}
\cos(2\rho + \beta)>-1+C\varepsilon^{1-\delta}, \quad \delta>0,
\end{equation}  
in which case we have indeed 
\begin{align} 
\frac{A(\varepsilon)}{B^2}& = 1 +\frac{\varepsilon}{B^2} + 
\bigl(1+ \frac{2\varepsilon}{B^2}\bigr)^{\frac12}\cos(2\rho + \beta) \notag \\
& > 1 +\frac{\varepsilon}{B^2} + 
\bigl(1+ \frac{\varepsilon}{B^2}-O(\varepsilon^2)\bigr)(-1 +C\varepsilon^{1-\delta}) 
\geq C\varepsilon^{1-\delta}\, .\label{lower-bound}
\end{align}

Let $\Omega \subset (0,1)$ be the set of $\varepsilon$ which does not satisfy 
\eqref{sufficient}. Then, Taylor's formula implies for some $C>0$ that 
\begin{equation} \label{f-st}
\Omega \subset \bigcup_{n=0}^\infty \bigl\{\varepsilon >0 \colon 
\bigl| \frac1{\varepsilon}-\frac{\pi}4 +  \beta -(2n+1)\pi 
\bigr|<C\varepsilon^{(1-\delta)/2}\bigr\}.
\end{equation} 
The definition of ${\beta}$ and Taylor's formula imply  
\[
\sin\beta= { \beta} - O({\beta}^3)= \frac{\sqrt{2\varepsilon}}{L_1}\bigl(1 
+ 
\frac{2\sqrt{2\varepsilon}}{L_1} + {\frac{4\varepsilon}{L_1^2}}\bigr)^{-\frac12} 
= \frac{\sqrt{2\varepsilon}}{L_1}\bigl(1 -  
\frac{\sqrt{2\varepsilon}}{L_1} + O(\varepsilon)\bigr) 
\]
and as $\varepsilon \to 0$   
\begin{equation} \label{2nd}
{\beta}= \frac{\sqrt{2\varepsilon}}{L_1} + O(\varepsilon). 
\end{equation} 
\eqref{2nd} implies that 
$\varepsilon>0$ which satisfies \eqref{f-st} must satisfy   
\begin{equation} \label{3rd}
\bigl| \frac1{\varepsilon}-\frac{\pi}4  - \frac{\sqrt{2\varepsilon}}{L_1} 
-(2n+1)\pi \bigr|<C\varepsilon^{(1-\delta)/2},
\end{equation} 
for some $n$ and (another) constant $C>0$. We want to solve \eqref{3rd} for 
$\varepsilon$ in terms of $n$. \eqref{3rd} is equivalent to 
\begin{equation} \label{4}
\bigl| \varepsilon -\frac1{\frac{\pi}4  -  \frac{2\sqrt{\varepsilon}}{L_1} +(2n+1)\pi} \bigr|<
\frac{C\varepsilon^{\frac{3-\delta}{2}}}
{\frac{\pi}4  - 
\frac{2\sqrt{\varepsilon}}{L_1} +(2n+1)\pi} .
\end{equation} 
For small $\varepsilon>0$ or for large $n$, this implies 
{ $(4n\pi)^{-1} \leq |\varepsilon|\leq C (n\pi)^{-1}$} and 
\[
\bigl| \varepsilon -\frac1{\frac{\pi}4 -
\frac{2\sqrt{\varepsilon}}{L_1}+(2n+1)\pi}\bigr|
< C n^{-(5-\delta)/2}. 
\]
Define 
\begin{equation}\label{b}
{ \varepsilon}(n)= \bigl(\frac{\pi}{4}+ (2n+1){ \pi} \bigr)^{-\frac12}.
\end{equation}
Then  
\[
|\sqrt{\varepsilon}- { \varepsilon(n)}|= 
\frac{|\varepsilon - {\varepsilon (n)}^2|}
{\sqrt{\varepsilon}+ { \varepsilon(n)}}
\leq Cn^{{ -\frac{3-\delta}{2}} }
\]
and 
\[
\bigl|\frac1{\frac{\pi}4 - \frac{2\sqrt{\varepsilon}}{L_1}+(2n+1)\pi}
- \frac1{\frac{\pi}4 -
\frac{2{ \varepsilon(n)}}{L_1}+(2n+1)\pi}\bigr| 
\leq Cn^{{ -\frac{7-\delta}{2}} }
\]
In this way we have shown that for some $C>0$ 
\begin{equation} \label{tildeW}
\Omega \subset \tilde{\Omega}=\bigcup_{n=0}^\infty 
\bigl\{\varepsilon \colon 
\bigl| \varepsilon -\frac1{
\frac{\pi}4 
 - \frac{2\varepsilon(n)}{L_1}
+(2n+1)\pi
}\bigr|
< C n^{-(5-\delta)/2} \bigr\}.
\end{equation}  
Thus we have obtained the following theorem. 

\begin{theorem} \label{thm2}
Let $0<\delta<1$, $B=L_1  -  \sqrt{2\varepsilon}$, $\varepsilon(n)$ be defined by \eqref{b}, 
and $\tilde{\Omega}$ by \eqref{tildeW} with a suitable constant $C>0$. 
Denote ${ \rho= \frac1{2\varepsilon} - \frac{\pi}{8}}$. 
Then, for $\varepsilon\notin\tilde{\Omega}$, 
$B^2+\varepsilon + B^2\cos2\rho  -  \sqrt{2\varepsilon}B \sin 2\rho\geq C \varepsilon^{1-\delta}$ 
and as $\varepsilon \to 0$  
\begin{equation}  
l_\varepsilon(1/{\varepsilon})= 
\frac{\varepsilon + i\left(B^2 \sin 2\rho  + 
\sqrt{2\varepsilon}{B}\cos 2\rho \right)+O(\varepsilon)}
{B^2+\varepsilon + B^2\cos2\rho  -  \sqrt{2\varepsilon}B \sin 2\rho}\left(1+ O(\varepsilon^\delta)\right), 
\end{equation} 
\end{theorem}

We notice that $\real l_\varepsilon(1/\varepsilon) \leq C\varepsilon^\delta$ 
and $|v_\varepsilon (1/\varepsilon,x)|\leq C_\varepsilon \exp(-C \varepsilon^\delta x^2/L)$ for $\varepsilon\not\in \tilde{\Omega}$. 
Recall that the free Schr\"odinger operator $-\partial_x^2$ has a zero resonance with resonant function 
$1$. It follows that, as $\varepsilon\not\in \tilde{\Omega}$ approaches $0$,  $v_\varepsilon (1/\varepsilon,x)$ approaches 
an oscillating function of the magnitude of the resonant function of $-\partial_x^2$ 
on every compact interval of $\mathbb{R}$, and it does so faster, when the length $L$ becomes 
longer, $2L$ being the time the particle stays as a free particle. Here the behavior 
as $\varepsilon \to 0$ of the imaginary part of $l_\varepsilon(1/\varepsilon)$ heavily depends on how $\varepsilon$ approaches $0$,  
however, we shall not pursue this point any further here.  

\paragraph{Acknowledgements}
A. Galtbayar was partially supported by the Asian Research Center at NUM, Grant No.2018-3589.
  A. Jensen was partially supported by the Danish 
Council for Independent Research $\mid$ Natural Sciences, Grant DFF--8021–00084B.
K. Yajima was partially supported by the JSPS grant in aid for scientific 
research No. 19K0358.


\end{document}